\documentclass{article}
\RequirePackage{amsthm,amsmath}
\usepackage{graphicx}

\usepackage{url}
\usepackage{rotating}
\usepackage{latexsym}
\usepackage{mathtools}
\usepackage{ascmac}
\usepackage{amssymb}
\usepackage{bm}

\theoremstyle{plain}
\newtheorem{theorem}{Theorem}
\theoremstyle{definition}
\newtheorem{definition}{Definition}

\newtheorem{corollary}{Corollary}
\theoremstyle{remark}
\newtheorem{example}{Example}
\newtheorem*{example2}{Example}

\newtheorem*{setting}{Setting}

\newcommand{\argmin}{\mathop{\rm arg~min}\limits}
\begin{document}

	\title{A proper scoring rule for minimum information copulas}
%	\runtitle{A scoring rule for minimum information copulas}

		\author{Yici Chen\thanks{Graduate School of Information Science and Technology, The University of Tokyo, 7-3-1 Hongo, Bunkyo-ku, Tokyo, 113-8656, Japan.}
		\and
		Tomonari Sei\footnotemark[1] \thanks{sei@mist.i.u-tokyo.ac.jp}}

\maketitle
		
	\begin{abstract}
		Multi-dimensional distributions whose marginal distributions are uniform are called copulas. Among them, the one that satisfies given constraints on expectation and is closest to the independent distribution in the sense of \rm{Kullback--Leibler} divergence is called the minimum information copula. The density function of the minimum information copula contains a set of functions called the normalizing functions, which are often difficult to compute.
		Although a number of proper scoring rules for probability distributions having normalizing constants such as exponential families are proposed, these scores are not applicable to the minimum information copulas due to the normalizing functions.
%       When performing estimation, a function that measures the goodness of the model, called the score function, is often used. Different scores have different properties, and by using scores with properties consistent with the intended use, estimation accuracy can be improved and the amount of computation can be reduced.
		In this paper, we propose the conditional Kullback--Leibler score, which avoids computation of the normalizing functions.
		The main idea of its construction is to use pairs of observations.
% as a kind of U-statistics.
		We show that the proposed score is strictly proper in the space of copula density functions and therefore the estimator derived from it has asymptotic consistency.
		%the error can be as close to the true value as possible if a sufficiently large amount of data is taken.
		Furthermore, the score is convex with respect to the parameters and can be easily optimized by the gradient methods.
	\end{abstract}
	
%	\begin{keyword}[class=MSC]
%		\kwd{62F10}
%		\kwd{62H05}
%%		\kwd[; secondary ]{60K35}
%	\end{keyword}

% Main text entry area

\section{Introduction}
Multidimensional distributions whose marginal distributions are uniform are called copulas. Among them, the one that satisfies the expectation constraints and is closest to the independent distribution in the sense of Kullback--Leibler (KL) divergence is called the minimum information copula. Minimum information copulas are used in, e.g., financial models \cite{CHATRABGOUN2018266} and flood models \cite{DANESHKHAH2016469}. It is known that the density function of a two-dimensional minimum information copula is written in the following form \cite{Bedford}:
\begin{align*}
&c_\theta(x,y)=\exp \left(\sum_{i=1}^k\theta_i h_i(x,y) +a_\theta(x)+b_\theta(y) \right),
\end{align*}
where $h_i(x,y)$ ($1\leq i\leq k$) are functions describing dependence between $x$ and $y$, and $\theta=(\theta_i)$ is a parameter.
%similar to the following exponential distribution family
The functions $a_\theta(x)$ and $b_\theta(y)$ are called normalizing functions that are determined by the marginal condition of copulas (see Section~\ref{section:min-info} for details).
The normalizing functions play a similar role to the normalizing constant of exponential families. However, since the marginal condition involves a set of integral equations, it is generally more difficult to deal with the normalizing functions than the normalizing constants.

When performing estimation, a function that measures the goodness of the model, called the score function, is often used.
Different scores have different properties, and by using scores with properties consistent with the intended use, estimation accuracy can be improved and the amount of computation can be reduced.
For example, scores for estimation that are robust to noise have been proposed by \cite{BHHJ} and \cite{FUJISAWA20082053} among others.
%In cases where it is difficult to calculate the normalizing constant of the model, 
The Hyv\"{a}rinen score \cite{hyvarinen2005estimation} has been proposed as a score without calculating the normalizing constant. A class of scoring rules with the same property are investigated by \cite{parry2012}.

In this paper, we propose a scoring rule for minimum information copulas that can be calculated without the normalizing functions.
The score, which we call the conditional Kullback--Leibler score, uses a conditional likelihood on pairs of observations.
The score is shown to be strictly proper and therefore asymptotically consistent. Furthermore, the score is convex with respect to the parameters and can be easily optimized by the gradient methods.

The structure of this paper is as follows. Section~2 introduces copulas and minimum information copulas. Section~3 defines scores and explains their commonly used properties: propriety, locality, and homogeneity. In Section~4, we introduce general homogeneity and multi-point locality, which are key properties for dealing with normalizing functions. Then, we propose the conditional Kullback--Leibler score
%a $2$-point local generally 0-homogeneous strictly proper score
for minimum information copulas that satisfies generalized homogeneity and two-point locality. In Section~5, we confirm the asymptotic consistency through numerical experiments. Finally, future issues and prospects are discussed in Section~6.

For simplicity, we will only discuss the two-dimensional case, but we believe it can be extended to three or higher dimensions. 

\section{Minimum information copulas} \label{section:min-info}
\subsection{Copulas}
A copula is a joint distribution function $C$ on the unit square with uniform marginals (e.g.\ \cite{Nels06}).
In this paper, we only deal with absolutely continuous copulas, in which case the definition is stated by density functions as follows.

\begin{definition}[Copula densities]
% 	A two-dimensional copula is a function $C:[0,1]^2 \to [0,1]$ with the following properties:
% 	\begin{enumerate}
% 		\item $C(x,0)=0=C(0,y) \hspace{2.6cm} \forall x,y \in [0,1]$;
% 		\item $C(x,1)=x$\; and \;$C(1,y)=y \hspace{1.4cm} \forall x,y \in [0,1]$;
% 		\item For all $ x_1,x_2,y_1,y_2 \in [0,1]$ such that $x_1 \leq x_2$ and $y_1 \leq y_2$,
% 		\begin{align*}
% 		C(x_2,y_2)-C(x_2,y_1)-C(x_1,y_2)+C(x_1,y_1) \geq 0.
% 		\end{align*}
% 	\end{enumerate}
	A two-dimensional copula density is a function $c:[0,1]^2 \to [0,\infty)$ that satisfies the following two properties:
	\begin{align}
	\int_0^1 c(x,y)\mathrm{d}y=1,\ \ x\in[0,1]
	\label{eq:marginal-1}
	\end{align}
	and
	\begin{align}
    \int_0^1 c(x,y)\mathrm{d}x=1,\ \ y\in[0,1].
    \label{eq:marginal-2}
	\end{align}
\end{definition}

From Sklar's theorem, a joint distribution with arbitrary marginals can be constructed by a copula.
Therefore, we can separate statistical modeling into the marginal part and the copula part.

\begin{theorem}[Sklar's theorem]
% 	Let $H$ be a joint distribution funtion with marginals $F$ and $G$. Then there exits a copula $C$ such that for all $x,y$ 
% 	\begin{align}
% 	H(x,y)=C(F(x),G(y)).
% 	\label{sklar}
% 	\end{align}
% 	If $F$ and $G$ are continuous, then $C$ is unique. Conversely, if $C$ is a copula and $F$ and $G$ are distribution functions, then the function $H$ defined by (\ref{sklar}) is a joint distribution function with marginals $F$ and $G$.
	Let $h$ be a joint density function on $\mathbb{R}^2$ with marginal density functions $f$ and $g$. Then there exists a copula density $c$ such that for all $(x,y)\in\mathbb{R}^2$,
	\begin{align*}
	h(x,y)=c(F(x),G(y))f(x)g(y),
%	\label{sklar}
	\end{align*}
	where $F(x)=\int_{-\infty}^x f(\xi)\mathrm{d}\xi$ and $G(y)=\int_{-\infty}^y g(\eta)\mathrm{d}\eta$.
\end{theorem}

\subsection{Minimum information copulas}
Consider statistical modeling of a phenomenon with a multivariate data.
Suppose that some constraints, such as mean and variance, are given as prior information.
If there is little prior information about the distribution, the model that satisfies the constraints cannot be uniquely determined. 
In this case, which model should be adopted? 
One way to think of it is to adopt a neutral model that assumes as little as possible dependence between random variables that are not included in the prior information.
In other words, we adopt the model that is closest to the distribution in which the random variables are independent.

The Kullback-Leibler(KL) divergence \cite{10.1214/aoms/1177729694} is often used to measure the distance between distributions.

\begin{definition}[Kullback--Leibler divergence]
	\label{KL-D}
	Let $g_1$ and $g_2$ be density functions. Then the Kullback--Leiber (KL) divergence is defined as follows:
	\begin{align*}
	D_{\rm{KL}}(g_1||g_2)=\int g_1(x,y)\log \left(\frac{g_1(x,y)}{g_2(x,y)} \right) \mathrm{d}x\mathrm{d}y.
	\end{align*}
\end{definition}

The copula that is closest to the independent copula in terms of the KL divergence is called the minimum information copula. 

\begin{definition}[Minimum information copulas]
    Let $h_1(x,y),\ldots,h_k(x,y)$ be given functions and $\alpha_1,\ldots,\alpha_k\in\mathbb{R}$ be given numbers.
    Let $\pi(x,y)=1$ be the independent copula density.
	Then, the copula density $c$ that minimizes
	\begin{align*}
	D_{\rm{KL}}(c||\pi)
	= \iint c(x,y)\log c(x,y)\mathrm{d}x\mathrm{d}y
	\end{align*}
	subject to
	\begin{align*}
	\iint c(x,y)h_i(x,y)\mathrm{d}x\mathrm{d}y=\alpha_i \ \ (1\leq i \leq k)
	\end{align*}
	is called the minimum information copula density.
\end{definition}

%For simplicity, we will only discuss the two-dimensional case, but we believe it can be easily extended to three or more dimensions. 
The minimum information copula is characterized by the following theorem. See also \cite{Borwein_et_al1994} for details on the uniqueness and existence problem.

\begin{theorem}[\cite{Bedford}, Theorem 2 \& Theorem 3]
	\label{min inf copula}
The minimum information copula density is unique if it exists, and expressed in the following form:
	\begin{align*}
	&c(x,y)=\exp \left(\sum_{i=1}^k\theta_i h_i(x,y) +a(x)+b(y) \right)
	\end{align*}
	with some $\theta_i$, $a(x)$ and $b(y)$. The functions $a(x)$ and $b(y)$ are unique except for arbitrariness of the additive constants.
\end{theorem}

From the theorem, we can redefine the minimum information copula density by
\begin{align}
&c_\theta(x,y)=\exp \left(\sum_{i=1}^k\theta_i h_i(x,y) +a_\theta(x)+b_\theta(y) \right)
\label{eq:min-info}
\end{align}
together with the marginal conditions (\ref{eq:marginal-1}) and (\ref{eq:marginal-2}).
 Here, the parameter of interest is $\theta=(\theta_i)_{i=1}^k$. The functions $a_\theta(x)$ and $b_\theta(y)$ are called the normalizing functions.

 For identifiability of $c_\theta(x,y)$, we suppose that $h_1(x,y),\ldots,h_k(x,y)$ are linearly independent modulo additive functions, which means that an identity
\[
 \sum_{i=1}^k \Theta_i h_i(x,y) + A(x) + B(y) = 0
\]
for some $\Theta_i$, $A(x)$ and $B(y)$ implies $\Theta_1=\cdots=\Theta_k=0$.

\section{Scoring rules} \label{section:scoring}
\subsection{Scores}
A score is a function used to measure the difference between a model and the true distribution.
We only consider probability density functions on $[0,1]^2$.
Denote the set of probability density functions on $[0,1]^2$ by $\mathcal{P}$.

\begin{definition}[Scores]
	A score $S$ is a real-valued function of $(x,y,q)\in[0,1]^2\times \mathcal{P}$. The expected value
	\begin{align*}
	S(p,q)=\iint S(x,y,q) p(x,y) \mathrm{d}x \mathrm{d}y
	\end{align*}
	with respect to $p\in\mathcal{P}$ is called the expected score.
\end{definition}

In this paper, the score function $S(x, y, q)$ and the expected score $S(p, q)$ use the same symbol $S$ and are both called scores for convenience. 
One property that is necessary for a score to measure the difference between the true distribution and the model is called propriety.

\begin{definition}[Propriety]
    A score $S$ is said to be proper if
	$S(p,q) \geq S(p,p)$
	for all density functions $p,q\in\mathcal{P}$.
\end{definition}

In general, when we say score, we often refer to the proper score. A divergence can be defined from a proper score by
%
% \begin{definition}[Divergences]
% 	Let $S$ be a proper score. Then, 
 	\begin{align*}
 	D(p||q)=S(p,q)-S(p,p),
 	\end{align*}
% 	is called the divergence of $S$.
%\end{definition}
%
%Let $p$ be the true distribution and $q$ be the model for a phenomenon.
which is like the distance between $p$ and $q$.
%If the model we set is the true distribution, $q = p$ and $D(p||q) = 0$.
%A divergence shifts the minimum value to 0. 
If $p$ is fixed, minimizing the divergence and minimizing the score are equivalent.  

Consider a statistical model $\{q_\theta\}\subset\mathcal{P}$ indexed by a parameter $\theta$.
If $n$ data points $(x_i , y_i )$ are observed and the empirical distribution is denoted as $\hat{p}$, the estimation can be operated as follows:
\begin{align}
\hat{\theta}%=\argmin_{\theta}D(\hat{p}||q_{\theta})
= \argmin_{\theta} S(\hat{p},q_{\theta})=\argmin_{\theta} \frac{1}{n}\sum^{n}_{i=1}S(x_i,y_i,q_{\theta}).
\label{eq:estimator}
\end{align}

There are many different types of scores, but one of the simplest is the local score.

% \begin{definition}[$\it{l}$-locality]
% 	Let $\it{l}$$>0$ be an integer and $\mathbb{Q}_l:=\mathbb{R}^+ \times \mathbb{R}^{\frac{1}{2}(l+2)(l+1)-1}$. Let $S:\mathbb{R}^2 \times \mathbb{Q}_{l} \rightarrow \mathbb{R}$ be a score.
% 	Let $s$ be a $C^{\infty}$ real value function and $q_{ij}:=\partial^i_x \partial^j_y \;q \;(i+j \leq l)$. Then, if score $S$ can be represented as
% 	\begin{align*}
% 	&S(x,y,q)=s\left( x,y,q_{00},q_{10},q_{01},\cdots,q_{ij},\cdots,q_{0l} \right)=s(x,y,\bm{q}) \\
% 	&\bm{q}:=(q_{00},q_{10},q_{01},\cdots,q_{ij},\cdots,q_{0l}),
% 	\end{align*}
% 	it is $\it{l}$-local and $\it{l}$ is the order of the score. ${0}$-local is called strictly local.
% \end{definition}
\begin{definition}[Locality \cite{parry2012}]
    A score $S$ is said to be local (in strict sense) if there exists a function $s:[0,1]^2\times[0,\infty)\to\mathbb{R}$ such that
    \[
    S(x,y,q)=s(x,y,q(x,y)).
    \]
    Furthermore, for $l\geq 0$, a score $S$ is said to be $l$-local if $S(x,y,q)$ is represented by at most the $l$-th derivatives of $q$ at $(x,y)$.
%There is the concept of $l$-locality for $l\geq 0$ as an extension of locality, which uses at most the $l$-th derivatives of $q$ (see \cite{parry2012}).
\end{definition}

Local scores are easy to calculate whenever $q(x,y)$ is explicitly expressed because the score can be obtained only from the information at that point, without integrating over the neighborhood or referring to other points.
%Since the set of strictly local scores is quite small and is not often used in discussions, $\it{l}$-locality is referred to as locality.

\begin{example}[\rm{KL} score] \label{example:KL}
    The score
	\begin{align*}
	S(x,y,q)=-\log q(x,y)
	\end{align*}
	is called the KL score because the divergence induced from it is the KL divergence.
The KL score is 0-local and proper.
In fact, it is known that the \rm{KL} score is essentially the only score that is 0-local and proper \cite{10.1214/aos/1176344689}.
\end{example}

\subsection{Homogeneity}
There are several computational advantages to using homogeneous scores for estimation.

% \begin{definition}[$\it{h}$-homogeneity]
% 	A score $S$ is said to be homogeneous of degree $h>0$, if, for any $c>0$, $S(x,y,cq)=c^h S(x,y,q)$.
% \end{definition}
\begin{definition}[Homogeneity]
	A score $S$ is said to be homogeneous if it satisfies $S(x,y,\lambda q)=S(x,y,q)$ for any constant $\lambda>0$.
\end{definition}

If a homogeneous score is used for estimation, computation of the normalizing constant is not necessary.

We have introduced the properties of scores: propriety, locality and homogeneity.
Based on these definitions, we consider two examples of scores.

\begin{example2}[Example~\ref{example:KL} continued]
The KL score is not homogeneous. Indeed, the KL score satisfies
\[
S(x,y,\lambda q)=S(x,y,q)-\log \lambda.
\]
Therefore, when the \rm{KL}-score is used for estimation, computation of the normalizing constant is necessary.
\end{example2}

%The second example is the \rm{Hyv\"{a}rinen} score.

\begin{example}[Hyv\"{a}rinen score \cite{hyvarinen2005estimation}]
A score
	\begin{align*}
	S(x,y,q)&=\left(\frac{\partial^2}{\partial x^2}+\frac{\partial^2}{\partial y^2}\right) \log q(x,y)\\
	&\ \ +\frac{1}{2}\left(\frac{\partial}{\partial x} \log q(x,y)\right)^2 +\frac{1}{2}\left(\frac{\partial}{\partial y} \log q(x,y)\right)^2
	\end{align*}
	is called the \rm{Hyv\"{a}rinen} score.
The \rm{Hyv\"{a}rinen} score is 2-local, homogeneous and proper. Therefore, the normalizing constant is not necessary for estimation.
However, the Hyv\"{a}rinen score is not useful for estimation of the minimum information copulas because it does not remove the normalizing functions.
\end{example}

%\section{A 2-points local generally 0-homoneneous strictly proper score}
\section{The proposed score}
\subsection{General homogeneity}

% From Theorem \ref{min inf copula}, if minimum information copulas exists, it can be represented in the following form:
% \begin{align*}
% &c(x,y)=\exp \left(\sum_{i}\theta_i h_i(x,y) +a(x)+b(y) \right).
% \end{align*}

As mentioned in the last example, the normalizing functions $a_\theta(x)$ and $b_\theta(y)$ in the density function (\ref{eq:min-info}) do not vanish even if a homogeneous score is applied.
For this reason, the following property is introduced.

\begin{definition}[General homogeneity]
	A score $S$ is said to be generally homogeneous if it satisfies that
	\begin{align*}
	&S(x,y,\lambda_1q)= S(x,y,q),\\
	&S(x,y,\lambda_2q)= S(x,y,q)
	\end{align*}
	for any positive functions $\lambda_1(x)$ and $\lambda_2(y)$,
	where $\lambda_1q$ and $\lambda_2q$ are defined by $(\lambda_1q)(x,y)=\lambda_1(x)q(x,y)$ and $(\lambda_2q)(x,y)=\lambda_2(y)q(x,y)$, respectively.
\end{definition}

\begin{example}
    It is easy to see that a score
    \[
    S(x,y,q)=-\frac{\partial^2}{\partial x\partial y}\log q(x,y)
    \]
    is generally homogeneous and 2-local. However, the score is not proper. To see this, let $q(x,y)$ be the Gaussian density (over $\mathbb{R}^2$).
    Then $S(x,y,q)$ is a constant involving the correlation parameter and $S(p,q)$ takes any real value independent of $p$.
\end{example}

If $S$ is generally homogeneous and $c_\theta(x,y)$ is the minimum information copula density in (\ref{eq:min-info}), we have
\[
 S(x,y,c_\theta(x,y))=S(x,y,e^{\sum_i\theta_ih_i(x,y)}),
\]
which does not require computation of the normalizing functions.
Hence, our problem is reduced to find a generally homogeneous ($l$-)local proper score.
However, after some trials based on symbolic computation in line with \cite{parry2012}, the authors realized that such a score may not exist.

%The reasons are explained in the supplementary materials. 
%supplementary materialを書くか迷ってます。本文に”拡張0-斉次局所適正スコアは存在しないと作者は思う”とだけ書いて終わるのも考えています。

\subsection{Multi-point scores and their locality}
Instead of finding a generally homogeneous local proper score, we try to relax the required properties.
General homogeneity and propriety are necessary for estimation of the minimum information copulas. 
Therefore, we reconsider locality. For this purpose, the concept of multi-point scores is introduced.

%So far, we have denoted scores as $S(x,y,q)$.
\begin{definition}[Multi-point score]
Let $m\geq 1$. An $m$-point score is a function
\[
S(\bm{x},\bm{y},q)
= S(x^1,y^1,\ldots,x^m,y^m,q)
\]
of $\bm{x}=(x^1,\cdots,x^m)\in[0,1]^m$, $\bm{y}=(y^1,\cdots,y^m)\in[0,1]^m$ and $q\in\mathcal{P}$. The expected score of $S(\bm{x},\bm{y},q)$ is defined as
	\begin{align*}
		S(p,q)&=\int S(\bm{x},\bm{y},q)p(\bm{x},\bm{y})\mathrm{d}\bm{x}\mathrm{d}\bm{y},
		\end{align*}
		where $p(\bm{x},\bm{y})= \prod_{i=1}^m p(x^i,y^i)$. In other words, the expectation is taken with respect to independent samples from $p$.
		A score is called a multi-point score if it is an $m$-point score for some $m$.
\end{definition}

Definition of propriety and general homogeneity of the multi-point scores is straightforward. Locality is defined as follows.

\begin{definition}[Multi-point locality]
    A multi-point score $S$ is said to be local if there exists a function $s:[0,1]^m\times [0,1]^m\times [0,\infty)^{m\times m}\to\mathbb{R}$ such that
    \[
     S(\bm{x},\bm{y},q) = s(\bm{x},\bm{y}, \bm{q}),
    \]
    where $\bm{q}=(q(x^i,y^j))_{i,j=1}^m$.
\end{definition}

Locality defined in Section~\ref{section:scoring} meant that for a given model, the score at a single point is evaluated using only the information at that point. On the other hand, multi-point locality means that the score at $m$ points is evaluated using all the information at the $m^2$ points $\{(x^i,y^j)\}_{i,j=1}^m$.
%Thus, a local score is multi-points ($1$-point) local, but a multi-points local score is not local in general.

%\subsection{A 2-point local generally homogeneous proper score}
\subsection{The conditional Kullback--Leibler score}
%Thus, if we weaken the locality, the set of scores will expand, and we may be able to make scores with desirable properties.
In the following, we construct a generally homogeneous $2$-point local proper score, which is applicable to estimation of minimum information copulas.
%This score is proper, not only for minimal information copulas, but also for all density functions.
Since it is known that the $1$-point local proper score is essentially only the \rm{KL}-score \cite{10.1214/aos/1176344689}, which does not have general homogeneity, it is reasonable to construct $2$-point local proper scores. 
%quite tight to find the original sense locality and propriety at the same time when making the scores. The reason why $\it{l}$-locality $\rm{\cite{parry2012}}$ was proposed was also to expand the set of scores.
%The multi-point locality proposed in this paper is also an extension because the set of scores with both $\it{l}$-locality and propriety could not produce scores with the other desired properties.

\begin{definition}[The conditional Kullback--Leibler score]
    Define a 2-point local score by
    \begin{align*}
		S(x^1,y^1,x^2,y^2,q)=-\log \left( \frac{q^{11} q^{22}}{q^{11}q^{22}+q^{12} q^{21}} \right),
	\end{align*}
	where $q^{ij}=q(x^i,y^j)$.
	We call it the conditional Kullback--Leibler score.
\end{definition}

This score looks a little strange at first glance, but it can actually be seen as a kind of conditional \rm{KL} score as follows.

Consider that there are two systems that are exactly the same and independent of each other.
Data $(x^1,y^1)$ are obtained from System~1, and data $(x^2,y^2)$ are obtained from System~2.
Suppose that, by mistake, we were able to record the numerical values of the data, but forgot the correspondence between $x$ and $y$.
Then, there are two possible cases: $(x^1,y^1),(x^2,y^2)$ or $(x^1,y^2),(x^2,y^1)$. 
Under the condition that only the numerical value of the data is known, the conditional probability that the data is the pair $(x^1,y^1),(x^2,y^2)$ is
	\begin{align*}
		\frac{q^{11} q^{22}}{q^{11}q^{22}+q^{12} q^{21}}.
	\end{align*}
% If data $(x^1,y^1)$ are obtained from System~1, and data $(x^2,y^2)$ are obtained from System~2, the \rm{KL}-score is 
% \begin{align*}
% 	S_{\rm{KL}}=-\log q^{11}q^{22}.
% \end{align*}
Therefore,
%if only the numerical value of the data is known, 
the score 
\begin{align*}
S=-\log \left( \frac{q^{11} q^{22}}{q^{11}q^{22}+q^{12} q^{21}} \right)
\end{align*}
will come naturally as the KL score for the conditional probability.

Now we state our main result.

\begin{theorem} \label{theorem:main}
    The conditional KL score is generally homogeneous and proper.
\end{theorem}

\begin{proof}[Proof]
    General homogeneity is straightforward.
    % \[
    % \frac{(\lambda_1^1q^{11})(\lambda_1^2q^{22})}{(\lambda_1^1q^{11})(\lambda_1^2q^{22})+(\lambda_1^1q^{12})(\lambda_1^2q^{21})} = \frac{q^{11} q^{22}}{q^{11}q^{22}+q^{12} q^{21}}
    % \]
    % where $\lambda_1^i=\lambda_1(x^i)$.
    We prove the propriety as follows:
 	\begin{align*}
        &S(p,p)-S(p,q)\\
        =&\int_{0}^{1}\int_{0}^{1}\int_{0}^{1}\int_{0}^{1}[S(x^1,y^1,x^2,y^2,p)-S(x^1,y^1,x^2,y^2,q)]p^{11}p^{22}\mathrm{d}x^1 \mathrm{d}y^1 \mathrm{d}x^2 \mathrm{d}y^2\\
        =&\frac{1}{2} \int_{0}^{1} \cdots \int_{0}^{1} [S(x^1,y^1,x^2,y^2,p)-S(x^1,y^1,x^2,y^2,q)]p^{11}p^{22}\mathrm{d}x^1 \mathrm{d}y^1 \mathrm{d}x^2 \mathrm{d}y^2 \\
        +&\frac{1}{2} \int_{0}^{1} \cdots \int_{0}^{1} [S(x^1,y^2,x^2,y^1,p)-S(x^1,y^2,x^2,y^1,q)]p^{12}p^{21}\mathrm{d}x^1 \mathrm{d}y^2 \mathrm{d}x^2 \mathrm{d}y^1 \\
        =&\frac{1}{2} \int_{0}^{1} \cdots \int_{0}^{1} \left[ \log \left(\frac{q^{11}q^{22}}{q^{11}q^{22}+q^{12}q^{21}}\right)-\log \left( \frac{p^{11}p^{22}}{p^{11}p^{22}+p^{12}p^{21}} \right) \right]p^{11}p^{22}\mathrm{d}x^1 \mathrm{d}y^1 \mathrm{d}x^2 \mathrm{d}y^2 \\
 	    +&\frac{1}{2} \int_{0}^{1} \cdots \int_{0}^{1} \left[ \log \left(\frac{q^{12}q^{21}}{q^{11}q^{22}+q^{12}q^{21}}\right)-\log \left( \frac{p^{12}p^{21}}{p^{11}p^{22}+p^{12}p^{21}} \right) \right]p^{12}p^{21}\mathrm{d}x^1 \mathrm{d}y^2 \mathrm{d}x^2 \mathrm{d}y^1 \\
        =&\frac{1}{2} \int_{0}^{1} \cdots \int_{0}^{1} \left[ \log \frac{\frac{q^{11}q^{22}}{q^{11}q^{22}+q^{12}q^{21}}}{\frac{p^{11}p^{22}}{p^{11}p^{22}+p^{12}p^{21}}} \right]\frac{p^{11}p^{22}}{p^{11}p^{22}+p^{12}p^{21}} (p^{11}p^{22}+p^{12}p^{21}) \mathrm{d}x^1 \mathrm{d}y^1 \mathrm{d}x^2 \mathrm{d}y^2 \\
        +&\frac{1}{2} \int_{0}^{1} \cdots \int_{0}^{1} \left[ \log \frac{\frac{q^{12}q^{21}}{q^{11}q^{22}+q^{12}q^{21}}}{\frac{p^{12}p^{21}}{p^{11}p^{22}+p^{12}p^{21}}} \right]\frac{p^{12}p^{21}}{p^{11}p^{22}+p^{12}p^{21}} (p^{11}p^{22}+p^{12}p^{21}) \mathrm{d}x^1 \mathrm{d}y^2 \mathrm{d}x^2 \mathrm{d}y^1.
     \end{align*}
    Using the inequality $\log x \leq x-1$, we obtain
	\begin{align*}
	&S(p,p)-S(p,q) \\
	\leq&\frac{1}{2} \int_{0}^{1} \cdots \int_{0}^{1} \left[ \frac{\frac{q^{11}q^{22}}{q^{11}q^{22}+q^{12}q^{21}}}{\frac{p^{11}p^{22}}{p^{11}p^{22}+p^{12}p^{21}}}-1 \right]\frac{p^{11}p^{22}}{p^{11}p^{22}+p^{12}p^{21}} (p^{11}p^{22}+p^{12}p^{21}) \mathrm{d}x^1 \mathrm{d}y^1 \mathrm{d}x^2 \mathrm{d}y^2 \\
	+&\frac{1}{2} \int_{0}^{1} \cdots \int_{0}^{1} \left[ \frac{\frac{q^{12}q^{21}}{q^{11}q^{22}+q^{12}q^{21}}}{\frac{p^{12}p^{21}}{p^{11}p^{22}+p^{12}p^{21}}}-1 \right]\frac{p^{12}p^{21}}{p^{11}p^{22}+p^{12}p^{21}} (p^{11}p^{22}+p^{12}p^{21}) \mathrm{d}x^1 \mathrm{d}y^2 \mathrm{d}x^2 \mathrm{d}y^1 \\
	=&\frac{1}{2} \int_{0}^{1} \cdots \int_{0}^{1} \left[ \frac{q^{11}q^{22}}{q^{11}q^{22}+q^{12}q^{21}}-\frac{p^{11}p^{22}}{p^{11}p^{22}+p^{12}p^{21}} \right] (p^{11}p^{22}+p^{12}p^{21}) \mathrm{d}x^1 \mathrm{d}y^1 \mathrm{d}x^2 \mathrm{d}y^2 \\
	+&\frac{1}{2} \int_{0}^{1} \cdots \int_{0}^{1} \left[ {\frac{q^{12}q^{21}}{q^{11}q^{22}+q^{12}q^{21}}}-{\frac{p^{12}p^{21}}{p^{11}p^{22}+p^{12}p^{21}}} \right](p^{11}p^{22}+p^{12}p^{21}) \mathrm{d}x^1 \mathrm{d}y^2 \mathrm{d}x^2 \mathrm{d}y^1 \\
	=&\frac{1}{2}\int_{0}^{1} \cdots \int_{0}^{1} \left[ \frac{q^{11}q^{22}}{q^{11}q^{22}+q^{12}q^{21}}-\frac{p^{11}p^{22}}{p^{11}p^{22}+p^{12}p^{21}}+{\frac{q^{12}q^{21}}{q^{11}q^{22}+q^{12}q^{21}}}-{\frac{p^{12}p^{21}}{p^{11}p^{22}+p^{12}p^{21}}} \right]\\
	&\hspace{8cm} \times (p^{11}p^{22}+p^{12}p^{21}) \mathrm{d}x^1 \mathrm{d}y^1 \mathrm{d}x^2 \mathrm{d}y^2 \\
	=&\frac{1}{2}\int_{0}^{1} \cdots \int_{0}^{1} \left[ 1-1 \right](p^{11}p^{22}+p^{12}p^{21}) \mathrm{d}x^1 \mathrm{d}y^1 \mathrm{d}x^2 \mathrm{d}y^2	\\
    =&0.
    \end{align*}
    Therefore, 
	\begin{align*}
        S(p,p)\leq S(p,q).
    \end{align*}
	As a result, the score $S$ is proper. Moreover, the equality condition is
	\begin{align}
	\frac{q^{11}q^{22}}{q^{11}q^{22}+q^{12}q^{21}}=\frac{p^{11}p^{22}}{p^{11}p^{22}+p^{12}p^{21}}
	\label{eq:equality-condition}
	\end{align}
	for all $(x^1,y^1,x^2,y^2)$.
\end{proof}

From the above, we construct a generally homogeneous $2$-point local proper score for the general density function. In fact, if the target is restricted to minimum information copulas, this score has even stronger properties.

\begin{definition}[Strict propriety]
    Let $\mathcal{M}\subset\mathcal{P}$ be a given class of probability density functions. A proper score $S$ is said to be strictly proper relative to $\mathcal{M}$ if the equality $S(p,q) = S(p,p)$ for $p,q\in\mathcal{M}$ implies $p=q$.
\end{definition}

\begin{theorem} \label{theorem:strictly-proper}
    Let $\mathcal{M}$ be a minimum information copula model.
    Then the conditional KL score is strictly proper relative to $\mathcal{M}$.
% 	For a minimum information copula $q(x,y)$, score 
% 	\begin{align*}
% 	S(x^1,y^1,x^2,y^2,q)=-\log \left( \frac{q^{11} q^{22}}{q^{11}q^{22}+q^{12} q^{21}} \right)
% 	\end{align*}
% 	is strictly proper. 
\end{theorem}

\begin{proof}[Proof]
%    Propriety has been shown in Theorem~\ref{theorem:main}.
%    We prove $S(p,q)=S(p,p)$ implies $p=q$.
    Let $p,q\in\mathcal{M}$ and suppose that $S(p,q)=S(p,p)$.
    Then we have (\ref{eq:equality-condition}) in the proof of Theorem~\ref{theorem:main},
%     we have \begin{align*}
% 		\frac{q^{11}q^{22}}{q^{11}q^{22}+q^{12}q^{21}}=\frac{p^{11}p^{22}}{p^{11}p^{22}+p^{12}p^{21}},
% 	\end{align*}
	which is equivalent to
	\begin{align*}
        \frac{q^{11}q^{22}}{q^{12}q^{21}}=\frac{p^{11}p^{22}}{p^{12}p^{21}},
	\end{align*}
	that is,
	\begin{align*}
		\frac{q(x^1,y^1)q(x^2,y^2)}{q(x^1,y^2)q(x^2,y^1)}=\frac{p(x^1,y^1)p(x^2,y^2)}{p(x^1,y^2)p(x^2,y^1)}.
	\end{align*}
	By fixing $(x^1,y^1)$ to an arbitrary point, we obtain a relation
	\[
	q(x^2,y^2) = p(x^2,y^2)\exp(a(x^2)+b(y^2)),
	\]
	where $a(x^2)=\log(p^{11}q^{21})-\log(q^{11}p^{21})$ and $b(y^2)=\log q^{12}-\log p^{12}$.
% 	Because $q,p$ are minimum information copulas, with the theorem $\ref{min inf copula}$, 
% 	\begin{align*}
% 		&q(x,y)=\exp\left\{ \sum_{i}\theta_i h_i(x,y)+a(x)+b(y)  \right\}=:\exp\left\{ H(x,y)+a(x)+b(y)  \right\}
% 	\end{align*}
% 	and
% 	\begin{align*}
% 		&p(x,y)=\exp\left\{ \sum_{i}\gamma_i k_i(x,y)+\alpha(x)+\beta(y) \right\}=:\exp \left\{K(x,y)+\alpha(x)+\beta(y)  \right\}.
% 	\end{align*}
% 	We consider that the parameters $\theta_i$ and $\gamma_i$ are included in the function $H(x.y)$ and $K(x,y)$ as constants.
% 	Therefore, with the equal condition, 
% 	\begin{align*}
% 		&\exp \left\{ H(x^1,y^2)+H(x^2,y^1)-H(x^1,y^1)-H(x^2,y^2) \right\}\\
% 		=&\exp \left\{ K(x^1,y^2)+K(x^2,y^1)-K(x^1,y^1)-K(x^2,y^2) \right\}\\
% 		\iff \\
% 		& H(x^1,y^2)+H(x^2,y^1)-H(x^1,y^1)-H(x^2,y^2)\\
% 		=& K(x^1,y^2)+K(x^2,y^1)-K(x^1,y^1)-K(x^2,y^2).
% 	\end{align*}
% 	With $x^1=x$, $y^1=y$, $x^2=0$, and $y^2=0$, 
% 	\begin{align*}
% 	& H(x,0)+H(0,y)-H(x,y)-H(0,0)\\
% 	=& K(x,0)+K(0,y)-K(x,y)-K(0,0).
% 	\end{align*}
% 	Operating $\partial_x \partial_y$ to each sides, 
% 	\begin{align*}
% 		\partial_x \partial_y H(x,y)=\partial_x \partial_y K(x,y) \hspace{1cm} (\forall x,y).
% 	\end{align*}
% 	Therefore, 
% 	\begin{align*}
% 		H(x,y)=K(x,y)+f_1(x)+f_2(y)+const \hspace{1cm} (\forall x,y).
% 	\end{align*}
    Since $p$ and $q$ are assumed to be the minimum information copula densities,
	Theorem~$\ref{min inf copula}$ implies
	\begin{align*}
		p(x^2,y^2)=q(x^2,y^2).
	\end{align*}
	Therefore, the score is strictly proper relative to $\mathcal{M}$.
\end{proof}

\subsection{Properties of the estimator}

We define an estimator based on the proposed score and briefly describe its properties.

For the conditional KL score, we first separate the given data randomly into $N=\lfloor n/2\rfloor$ groups as
\[
 \{(x_i^1,y_i^1,x_i^2,y_i^2)\}_{i=1}^N.
\]
Then, based on the empirical score
\[
 \hat{S}(\theta) =\frac{1}{N}\sum_{i=1}^{N}S(x^1_i,y^1_i,x^2_i,y^2_i,q),
\]
the estimator is defined by
\begin{align*}
 \hat\theta = \argmin_{\theta}\hat{S}(\theta)
\end{align*}
%Optimization of $\hat{S}$ can be easily done using the gradient method. 

Recall the following theorem on consistency and asymptotic normality of estimators based on strictly proper scores. We omit regularity conditions to make the ideas clearer.

\begin{theorem}[\cite{doi:10.1198/016214506000001437} and Theorem~5.23 of \cite{van2000asymptotic}]
	\label{strictly_proper}
	Let $\theta_0$ be the true parameter and suppose that $(x_1,y_1),\cdots,(x_n,y_n)$ are independent and identically distributed. Let $S$ be a strictly proper (1-point) score.
	Then, the estimator $\hat\theta$ defined by (\ref{eq:estimator}) converges almost surely to $\theta_0$ as $n\to\infty$.
% 	\begin{align*}
% 		\hat\theta = \argmin_{\theta} \frac{1}{n}\sum_{i=1}^{n}S(x_i,y_i,q_{\theta}) \xrightarrow{\rm a.s.} \theta_0
% 	\end{align*}
	Furthermore, under regularity conditions, the asymptotic normality holds:
	\[
	\sqrt{n}(\hat\theta-\theta_0) \xrightarrow{\rm d} N(0,J^{-1}VJ^{-1}),
	\]
	where $J=E[\nabla_\theta\nabla_\theta^\top S]$ and $V=E[(\nabla_\theta S)(\nabla_\theta S)^\top]$.
\end{theorem}

Since the conditional KL score is strictly proper as proved in Theorem~\ref{theorem:strictly-proper}, we obtain the following corollary.

\begin{corollary} \label{corollary:asymptotic}
 The estimator based on the conditional KL score is consistent and asymptotically normal.
\end{corollary}

% \begin{align*}
% 	S(x^1,y^1,x^2,y^2,q)=-\log \left( \frac{q^{11} q^{22}}{q^{11}q^{22}+q^{12} q^{21}} \right)
% \end{align*}
% is $2$-points local generally $0$-homogeneous strictly proper. 
% If we take enough data points in the parameter estimation, it converges to the true parameters.

Next we point out that the estimation is a convex optimization problem.
Here, we use the vector notation $\bm\theta=(\theta_i)_{i=1}^k$ for convenience.

\begin{theorem}
    Consider a minimum information copula model
    \[
		q(x,y;\bm{\theta})=\exp\left( \bm{\theta}^{T} \bm{h}(x,y)+a(x)+b(y) \right).
	\]
	Suppose that $\bm{H}_1,\ldots,\bm{H}_N\in\mathbb{R}^k$ are linearly independent, where
	\[
		\bm{H}_i=\bm{h}_i^{12}+\bm{h}_i^{21}-\bm{h}_i^{22}-\bm{h}_i^{11},
		\quad \bm{h}_i^{\alpha\beta}=\bm{h}(x_i^\alpha,y_i^\beta).
	\]
	Then, the empirical score $\hat{S}(\bm\theta)$ based on the conditional KL score
	is strictly convex with respect to $\bm{\theta}$.
\end{theorem}

\begin{proof}[Proof]
%    Let $q_i^{kl}=q(x_i^k,y_i^l)$.
	It is easy to see
	\begin{align*}
	 \hat{S}
	 &= \frac{1}{N}\sum_{i=1}^{N} \left\{-\log\left(
	 \frac{q_i^{11}q_i^{22}}{q_i^{11}q_i^{22}+q_i^{12}q_i^{21}}
	 \right)\right\}
	 \\
	 &= \frac{1}{N}\sum_{i=1}^{N} \log \left( \mathrm{e}^{\bm{\theta}^{T} \bm{H}_i}+1 \right).
	\end{align*}
	Then, the gradient vector is
	\begin{align}
	\nabla_{\bm\theta} \hat{S}
	&=\frac{1}{N}\sum_{i=1}^{N}\frac{\mathrm{e}^{\bm{\theta}^{T} \bm{H}_i}}{\mathrm{e}^{\bm{\theta}^{T} \bm{H}_i}+1} \bm{H}_i \nonumber \\
	&=\frac{1}{N}\sum_{i=1}^{N}\frac{\bm{H}_i}{\mathrm{e}^{-\bm{\theta}^{T}  \bm{H}_i}+1} \label{partial_theta_S}.
	\end{align}
	The Hessian matrix
	\begin{align*}
	\frac{1}{N}\sum_{i=1}^{N}\frac{\mathrm{e}^{-\bm{\theta}^{T}  \bm{H}_i}}{(\mathrm{e}^{-\bm{\theta}^{T}  \bm{H}_i}+1)^2}\bm{H}_i \bm{H}_i^T,
	\end{align*}
	is positive definite under the assumption.
	Therefore, the score is strictly convex with respect to $\bm\theta$.
\end{proof}

% Therefore, the score
% \begin{align*}
% S(x^1,y^1,x^2,y^2,q)=-\log \left( \frac{q^{11} q^{22}}{q^{11}q^{22}+q^{12} q^{21}} \right)
% \end{align*}
% is $2$-points local generally $0$-homogeneous strictly proper.
The estimation can be operated easily with the gradient method. 

%Therefore, this score is often referred to as the \rm{CKL (Conditional Kullback-Leibler)}-score in this paper.

\section{Numerical experiments}
\subsection{Experiments of Gaussian copulas}
\label{gaussian_copula_exp}
\subsubsection{Setting}
\label{setting}
The normalizing function of a minimal information copula cannot be obtained in general.
In this subsection, we consider the Gaussian copula, which is one of the few examples of minimal information copulas for which the normalization function can be obtained in a  closed form.

Another reason for experimenting with a Gaussian copula is that the data can be easily generated by a variable transformation of the $2$-dimensional normal distribution, and the maximum likelihood estimation (MLE) can be compared with the proposed method.
In Subsection \ref{general numerical experiments}, we conduct numerical experiments for general minimum information copulas for which normalizing functions cannot be obtained.

The Gaussian copula is a multidimensional normal distribution with its marginal distributions converted to uniform distributions.
%\begin{definition}[$2$-dimensional Gaussian copulas]
	The density function of the $2$-dimensional Gaussian copula is 
	\begin{align*}
%		&c(x,y|\theta)=
		\frac{1}{2 \pi (1-\rho^2)^{\frac{1}{2}}\phi(\xi)\phi(\eta)}\exp \left. \left( - \frac{1}{2(1-\rho^2)} (\xi^2-2 \rho \xi \eta + \eta^2)  \right) \right|_{\xi=\Phi^{-1}(x),\eta=\Phi^{-1}(y)}
	\end{align*}
	for $(x,y)\in(0,1)^2$, where $\phi$ and $\Phi$ are the density function and cumulative distribution function of the $1$-dimensional standard normal distribution. The parameter $\rho$ is the correlation coefficient of the Gaussian variables $\xi$ and $\eta$.

	As pointed out by \cite{jansen1997maximum}, the Gaussian copula is in fact a minimum information copula with 
	\begin{align*}
		\theta=\frac{\rho}{1-\rho^2}
	\end{align*}
	and
	\begin{align*}
		&h(x,y)=\Phi^{-1}(x)\Phi^{-1}(y).
	\end{align*}
	The normalizing functions are
	\begin{align*}
		&\exp\{a(x)\}=\frac{1}{\sqrt{2\pi} (1-\rho^2)^{\frac{1}{4}}\phi(\xi)}\exp \left. \left( -\frac{\xi^2}{2(1-\rho^2)} \right) \right|_{\xi=\Phi^{-1}(x)}
	\end{align*}
	and $\exp\{b(y)\}=\exp\{a(y)\}$.
% 	\begin{align*}
% 		&\exp\{b(y)\}=\frac{1}{\sqrt{2\pi} (1-\rho^2)^{\frac{1}{4}}\phi(\eta)}\exp \left. \left( -\frac{\eta^2}{2(1-\rho^2)} \right) \right|_{\eta=\Phi^{-1}(y)}.
% 	\end{align*}
	The parameters $\theta$ and $\rho$ have one-to-one correspondence as
	\[
	\rho=\frac{2\theta}{1+\sqrt{1+4\theta^2}}.
    \]
%\end{definition}

Thus, the density function of the Gaussian copula can be expressed relatively simply. In the following, the procedure of the numerical experiment is explained in detail.

\begin{setting}
	\begin{enumerate}
		\item First, we generate $2N$ data of $2$-dimensional normal distribution with mean vector$
		\begin{bmatrix}
		0 \\
		0
		\end{bmatrix}
		$and covariance matrix$
		\begin{bmatrix}
		1    & \rho \\
		\rho & 1
		\end{bmatrix}.
		$
		Consider the first $N$ data$
		\begin{bmatrix}
		\xi_i^1 \\
		\eta_i^1
		\end{bmatrix}
		$
		as obtained from $\text{system}$ $1$ and the other $N$ data$
		\begin{bmatrix}
		\xi_i^2 \\
		\eta_i^2
		\end{bmatrix}
		$
		as obtained from $\text{system}$ $2$.
		\item 
% 		Let error function
% 		\begin{align*}
% 		\text{erf}(x)=\frac{2}{\sqrt{\pi}}\int_{0}^{x} \text{e}^{-t^2} \mathrm{d}t.
% 		\end{align*}
% 		With the sample of $2$-dimensional normal distribution$
% 		\begin{bmatrix}
% 		\xi_i^j \\
% 		\eta_i^j
% 		\end{bmatrix}
% 		(j=1,2)
% 		$
		Take the variable transformation
		\begin{align}
		\label{transform}
		\begin{bmatrix}
		x_i^j \\
		y_i^j 
		\end{bmatrix}
		=\begin{bmatrix}
		\Phi(\xi_i^j)\\
		\Phi(\eta_i^j)
		\end{bmatrix}
		=\begin{bmatrix}
		\frac{1}{2}\left(1+\text{erf}\left( \frac{\xi_i^j}{\sqrt{2}} \right)\right) \\
		\frac{1}{2}\left(1+\text{erf}\left( \frac{\eta_i^j}{\sqrt{2}} \right)\right)
		\end{bmatrix}
		\end{align}
		($i=1,\ldots,N;j=1,2$)
		in order to obtain $2N$ sample of the Gaussian copula, where $\text{erf}(x)=(2/\sqrt{\pi})\int_{0}^{x} \text{e}^{-t^2} \mathrm{d}t$ is the error function.
		\item 
		%Estimate the parameter $\theta=\frac{\rho}{1-\rho^2}$ with the proposed score.
		Set the initial value $\theta=\theta_0$ and iterative step $\mathrm{d}\theta$ properly. 
		\item Calculate the gradient of the empirical score
		\begin{align*}
		    \hat{S}:=\frac{1}{N}\sum_{i=1}^{N}S(x_i^1,y_i^1,x_i^2,y_i^2,q)=-\frac{1}{N}\sum_{i=1}^{N}\log \left( \frac{q^{11}_{i}q^{22}_{i}}{q^{11}_{i}q^{22}_{i}+q^{12}_{i}q^{21}_{i}} \right).
        \end{align*}
Substitute the sample $(x_i^j,y_i^j)(i=1,\cdots,N;j=1,2)$ for $(\ref{partial_theta_S})$:
        \begin{align*}
            \frac{\mathrm{d}}{\mathrm{d}\theta}\hat{S}=\frac{1}{N}\sum_{i=1}^{N}\frac{{H}_i}{\mathrm{e}^{-\theta {H}_i}+1}.
        \end{align*}
%			  and calculate the iterative direction $-\frac{\mathrm{d}\hat{S}}{\mathrm{d}\theta}$.
		\item 
		%Stop the iteration $\theta_{i+1}=\theta_{i}-\frac{\mathrm{d}\hat{S}}{\mathrm{d}\theta}\mathrm{d}\theta$ when the gradient $\frac{\mathrm{d}\hat{S}}{\mathrm{d}\theta}$ is sufficiently near to $0$. 
		If $|\mathrm{d}\hat{S}/\mathrm{d}\theta|$ is sufficiently small, output $\theta$ as the estimator $\hat\theta$. Otherwise, $\theta\leftarrow \theta-(\mathrm{d}\hat{S}/\mathrm{d}\theta)\mathrm{d}\theta$ and go to Step 4.
		%\item $\theta_i$ that makes $\frac{\mathrm{d}\hat{S}}{\mathrm{d}\theta}$ sufficiently near to $0$, is the estimator $\hat{\theta}$.
	\end{enumerate}
\end{setting}

The experiments of the MLE  use the empirical score of the \rm{KL}-score
\begin{align*}
	\hat{S}_{\rm{KL}}=-\frac{1}{N}\sum_{i=1}^{N}\log{\left(q^{11}_{i}q^{22}_{i}\right)}
\end{align*}
instead of the conditional KL score in Step $4$. 
The other steps are the same.

\subsubsection{Results}
\label{result}
We describe the results of numerical calculations according to the experimental procedure in Subsection $\ref{setting}$.

First, $4000$ sample points of $2$-dimensional normal distribution with mean vector$
\begin{bmatrix}
0 \\
0
\end{bmatrix}
$ and covariance matrix$
\begin{bmatrix}
1    & \rho=0.7 \\
\rho=0.7 & 1
\end{bmatrix}
$
were obtained.
Then, the variable transformation $(\ref{transform})$ was operated in order to get the sample of Gaussian copula with parameter $\theta=\frac{\rho}{1-\rho^2}=1.372549$.
The larger the number of data, the more the estimation error converges to $0$. To see the degree of convergence, the estimation error was calculated by changing the number of data used.
Using $N=40,50,\cdots,1990,2000$ pairs of data, we substituted the data into the proposal score and obtained the optimal solution $\hat{\theta}$ as the estimator.
The estimation error was calculated as the absolute value of the difference from the true parameter $\theta=1.372549$.
In addition, maximum likelihood estimation was performed using the same data with \rm{KL}-scores.
The above experiment was repeated up to $100$ times and the estimation errors were averaged.

The estimation results are shown in Figure~\ref{Gaussian_copula_error_MLE_CKL}.
The results are also plotted on the logarithmic graph of both axes and linearly fitted in Figure~\ref{Gaussian_copula_log_error_MLE_CKL}.
Linear fitting is the least-squares fitting of $a,b$ of $y=x^a \exp(b)$ to the data.
The red line is the estimation error of the maximum likelihood estimation (\rm{KL}-score) and the green is the estimation error of the proposal score (\rm{CKL}-score).
The blue line is the estimation error fitting of the maximum likelihood estimation (\rm{KL}-score), where the coefficients of the fitting are
\begin{align*}
&a=-0.517919\\
&b=0.445423.
\end{align*}
The purple line is the estimation error fitting of the proposed score (\rm{CKL}-score), where
\begin{align*}
&a=-0.49438\\
&b=0.97278.
\end{align*}
The speed of convergence of the error is roughly $\frac{1}{\sqrt{N}}$ for the number $N$ of data since the blue and purple $a$ in the fitting are close to $-0.5$.
This result is consistent with Corollary~\ref{corollary:asymptotic}.
The maximum likelihood estimation has better results with respect to error convergence than the proposed method, as expected.
However, when estimating parameters with minimum information copulas, the maximum likelihood estimation is only possible with the Gaussian copulas used in this experiment.
In most other cases, the maximum likelihood estimation cannot be performed because the normalizing functions cannot be obtained.
In such cases, the greatest strength of the proposed score is that it can estimate with the same accuracy.

\begin{figure}[htbp]
	\centering
	\includegraphics[width=10cm,height=6.5cm]{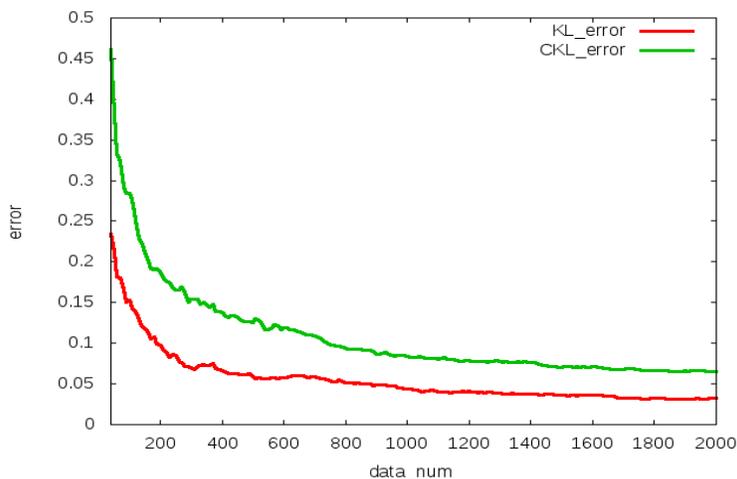}
	\caption{The estimation result of Gaussian copulas. The mean absolute error (vertical axis) with respect to the number of data (horizontal axis) is shown. The red line is the KL score (MLE) and green line is the CKL score (proposed)}
	\label{Gaussian_copula_error_MLE_CKL}
\end{figure}

\begin{figure}[htbp]
	\centering
	\includegraphics[width=10cm,height=6.5cm]{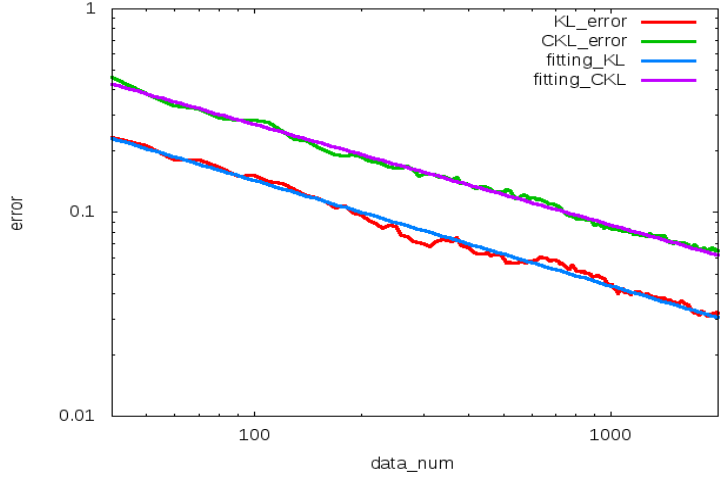}
	\caption{The estimation result of Gaussian copulas. The mean absolute error (vertical axis) with respect to the number of data (horizontal axis) is shown in logarithmic scale. The red line is the KL score (MLE), green line is the CKL score (proposed), and the blue and purple lines are the linear fitting of the two lines.}
	\label{Gaussian_copula_log_error_MLE_CKL}
\end{figure}

\subsection{Experiments of general minimum information copulas}
\label{general numerical experiments}
\subsubsection{Setting}

% In this experiment, MLE is not possible, so the results are based on the proposed method only. In addition, since sampling itself is not easy, we generate data from the approximate sampling method \cite{Sei&Yano}.

Section \ref{gaussian_copula_exp} described the numerical experimental results of the Gaussian copulas. In this section, we present numerical experiments on general minimum information copulas. Since the normalizing functions of the minimum information copula are not generally obtainable, exact sampling is difficult. Thus, the generation of data itself is problematic before parameter estimation. In this paper, we generate data using the approximate sampling method proposed by \cite{Sei&Yano}. The sampling procedure is described below.
%For more details, please refer to \cite{Sei&Yano}.

\begin{setting}[The approximate sampling method]
    \begin{enumerate}
		\item First, we decide the function $h(x,y)$ and parameter $\theta \in \mathbb{R}$. Once these are determined, the minimum information copula is uniquely determined, although it cannot be written explicitly.
		\item Then, generate a sufficiently large number of independent and identically distributed (i.i.d.) data $
		\begin{bmatrix}
		x_i \\
		y_i
		\end{bmatrix}(i=1,\cdots,2N)$ from the uniform distribution on $[0,1]^2$.
		\item Choose $i,j$ that satisfy $1\leq i <j \leq 2N$ randomly and flip them with probability
		\begin{align*}
			\rho
			&=\frac{q(x_i,y_j)q(x_j,y_i)}{q(x_i,y_i)q(x_j,y_j)+q(x_i,y_j)q(x_j,y_i)}\\
			&=\frac{\exp\{ \theta(h(x_i,y_j)+h(x_j,y_i)) \}}{\exp \{ \theta(h(x_i,y_i)+h(x_j,y_j))+\exp\{\theta (h(x_i,y_j)+h(x_j.y_i)) \}}
		\end{align*}
		In other words, change the pair $
		\begin{bmatrix}
		x_i \\
		y_i
		\end{bmatrix}$ and $
		\begin{bmatrix}
		x_j \\
		y_j
		\end{bmatrix}$ to the pair$
		\begin{bmatrix}
		x_i \\
		y_j
		\end{bmatrix}$ and $
		\begin{bmatrix}
		x_j \\
		y_i
		\end{bmatrix}$.
		\item Repeat step $3$ enough times. Then, the $2N$ data are approximately an i.i.d.\ sample from the minimum information copula decided in step $1$.
		\item We consider the first $N$ data$
		\begin{bmatrix}
		x_i^1 \\
		y_i^1
		\end{bmatrix}$
		obtained from system $1$ and the other $N$ data $
		\begin{bmatrix}
		x_i^2 \\
		y_i^2
		\end{bmatrix}
		$obtained from system $2$.
	\end{enumerate}
\end{setting}

Now we have an approximate sample from the given minimum information copulas and can perform parameter estimation using the proposed scores.
Once the initial value $\theta_0$ and the iterative step $\mathrm{d}\theta$ are determined appropriately, the rest is done in the same way as in step $4$ and $5$ of Subsection~\ref{setting}. The above numerical experiments are repeated up to $100$ times from sampling to estimation, and the average of the absolute errors is obtained.

\subsubsection{Result}
First, to roughly check the behavior of the approximate sampling, we conduct the experiment with Gaussian copulas again. 
%For a detailed description of the sampling technique, please refer to \cite{Sei&Yano}.
The exact and approximate sampling were used, respectively, to obtain a sample of Gaussian copulas with the same parameters $\theta=\frac{\rho}{1-\rho^2}=1.372549$ as in Subsection~\ref{result}.
Using $N=40,50,\cdots,1990,2000$ pairs of data, we substituted the data into the proposed score and obtained the optimal solution $\hat{\theta}$ as the estimator.
The estimation error was calculated as the absolute value of the difference from the true parameter $\theta=1.372549$.
In addition, maximum likelihood estimation was performed using the same data with \rm{KL}-scores.
The above experiment was repeated up to $100$ times and the estimation errors were averaged.

% The results of estimating data with variable-transformed $2$-dimensional normal distribution and approximate sampling are shown in Figure\ref{Gaussian_copula_error}.
% The red line is the estimation error for sampling with two-dimensional normal distributions and the maximum likelihood estimation (\rm{KL}-score), and the green line is the estimation error for sampling with two-dimensional normal distributions and proposed score (\rm{CKL}-score).
% The blue line is the estimation error of approximate sampling and the maximum likelihood estimation (\rm{KL}-score), and the purple line is the estimation error of approximate sampling and proposed score (\rm{CKL}-score).

% \begin{figure}[htbp]
% 	\centering
% 	\includegraphics[width=10cm,height=6.5cm]{Gaussian_copula_error.png}
% 	\caption{The estimation of Gaussian copulas: \newline
% 		Comparison of sampling by the $2$-dimensional normal distribution and approximate sampling\newline
% 		Variation of absolute error (vertical axis) with respect to the number of data (horizontal axis)\newline
% 		red: sampling with the $2$-dimensional normal distribution $+$\rm{KL}-score(MLE)\newline
% 		green: sampling with the $2$-dimensional normal distribution $+$proposed score\newline
% 		blue: approximate sampling $+$\rm{KL}-score(MLE)\newline
% 		purple: approximate sampling $+$proposed score}
% 	\label{Gaussian_copula_error}
% \end{figure}

The results of estimating data with exact and approximate sampling are shown in 
Figure~\ref{Gaussian_copula_log_error}, where both axes are in logarithmic scale.
%Re-plotted on a logarithmic graph on both axes is, %and linearly fitted graph is Figure \ref{Gaussian_copula_log_error_fitting}.
The red line shows the estimation error for the exact sampling and the maximum likelihood estimation (\rm{KL}-score), and the green line is the estimation error for the exact sampling and the proposed score (\rm{CKL}-score).
The blue line is the estimation error of the approximate sampling and the maximum likelihood estimation (\rm{KL}-score), and the purple line is the estimation error of the approximate sampling and proposed score (\rm{CKL}-score).

\begin{figure}[htbp]
	\centering
	\includegraphics[width=10cm,height=6.5cm]{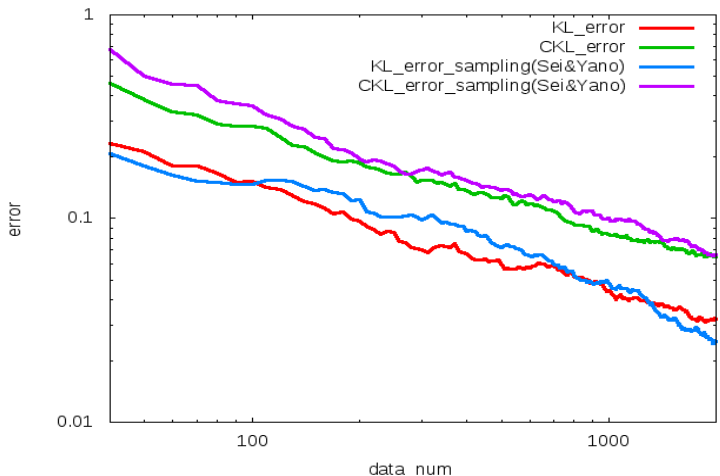}
	\caption{The estimation result of Gaussian copulas with exact and approximate sampling. The mean absolute error (vertical axis) with respect to the number of data (horizontal axis) is shown in logarithmic scale. The four lines mean the exact sampling $+$ KL score (red), the exact sampling $+$ CKL score (green), the approximate sampling $+$ KL score (blue) and the approximate sampling $+$ CKL score (purple).}
	\label{Gaussian_copula_log_error}
\end{figure}

%Linear fitting is the least-squares fitting of $a,b$ of $y=x^a \exp(b)$ to the data.
The result of linear fitting for the four lines is summarized in Table~\ref{table:linear-fitting-Gaussian}.
The figure and table confirm that the approximate sampling seems successfully generating data for the Gaussian copulas.

\begin{table}[htbp]
\caption{The coefficients of linear fitting for the four lines on in Figure~\ref{Gaussian_copula_log_error}.}
\label{table:linear-fitting-Gaussian}
 \begin{center}
 \begin{tabular}{cc|cc}
 score& sampling & $a$& $b$ \\
 \hline
KL & exact &$-0.517919$& 0.445423\\
CKL & exact & $-0.49438$& 0.97278\\
KL & approximate& $-0.492425$& 0.37061\\
CKL & approximate& $-0.561812$& 1.52274
\end{tabular}
\end{center}
\end{table}

% The red line is the error of sampling due to the two-dimensional normal distribution and maximum likelihood estimation (\rm{KL}-score), and the result of linear fitting is
% \begin{align*}
% &a=-0.517919\\
% &b=0.445423.
% \end{align*}
% The green line is the error of sampling due to the two-dimensional normal distribution and proposed score (\rm{CKL}-score), and the result of linear fitting is
% \begin{align*}
% &a=-0.49438\\
% &b=0.97278.
% \end{align*}
% The blue line is the error of approximate sampling and maximum likelihood estimation (\rm{KL}-score), and the result of linear fitting is
% \begin{align*}
% &a=-0.492425\\
% &b=0.37061.
% \end{align*}
% The purple line is the erro of approximate sampling and proposed score (\rm{CKL}-score), and the result of linear fitting is
% \begin{align*}
% &a=-0.561812\\
% &b=1.52274.
% \end{align*}

Since the rough behavior of approximate sampling has been confirmed, we will discuss the results of experiments on sampling and parameter estimation of general minimum information copulas, which was the original purpose of this paper.

We generated $N=2000$ pairs of samples of minimum information copula with parameters $\theta=5.0,10.0$ and function $h(x,y)=xy,x^2y$. We used $N=20,30,40,\cdots,1990,2000$ pairs of them, substituted data into the proposed score, and estimated the optimal solution as $\hat{\theta}$. The estimation error was calculated as the absolute value of the difference from the true parameter $\theta$.
The above experiment was repeated up to $100$ times and the average of the estimation error was taken.

The estimation results with the proposed score for the general minimum information copulas are shown in %Figure\ref{general_copula_error}.
Figure~\ref{general_copula_log_error}, where both axes are in logarithmic scale.
The red line is the estimation error for $\theta=5.0,h(x,y)=xy$ and the green line is the estimation error for $\theta=5.0,h(x,y)=x^2y$.
The blue line is the estimation error for $\theta=10.0,h(x,y)=xy$ and the purple is the estimation error for $\theta=5.0,h(x,y)=x^2y$.

% \begin{figure}[htbp]
% 	\centering
% 	\includegraphics[width=10cm,height=6.5cm]{Gaussian_copula_log_error_fitting.png}
% 	\caption{The estimation of Gaussian copulas (logarithmic graph on both axes): \newline
% 		Comparison of sampling by2-dimensional normal distribution vs. approximate sampling\newline
% 		Variation of absolute error (vertical axis) with respect to the number of data (horizontal axis)\newline
% 		red: The linear fitting of the error with $2$-dimensional normal distribution sampling $+$\rm{KL}-score(MLE)\newline
% 		green: The linear fitting of the error with $2$-dimensional normal distribution sampling $+$proposed score\newline
% 		blue: The linear fitting of the error with approximate sampling $+$\rm{KL}-score(MLE)\newline
% 		purple: The linear fitting of the error with approximate sampling $+$proposed score}
% 	\label{Gaussian_copula_log_error_fitting}
% \end{figure}

% \begin{figure}[htbp]
% 	\centering
% 	\includegraphics[width=10cm,height=7cm]{general_copula_error.png}
% 	\caption{The estimation of general minimum information copulas: \newline
% 		Variation of absolute error (vertical axis) with respect to the number of data (horizontal axis)\newline
% 		red: $\theta=5.0,h(x,y)=xy$\newline
% 		green: $\theta=5.0,h(x,y)=x^2y$\newline
% 		blue: $\theta=10.0,h(x,y)=xy$\newline
% 		purple: $\theta=10.0,h(x,y)=x^2y$
% 	}
% 	\label{general_copula_error}
% \end{figure}

% Re-plotted on a logarithmic graph on both axes is Figure \ref{general_copula_log_error}, and linearly fitted graph is Figure \ref{general_copula_log_error_fitting}.

\begin{figure}[htbp]
	\centering
	\includegraphics[width=10cm,height=7cm]{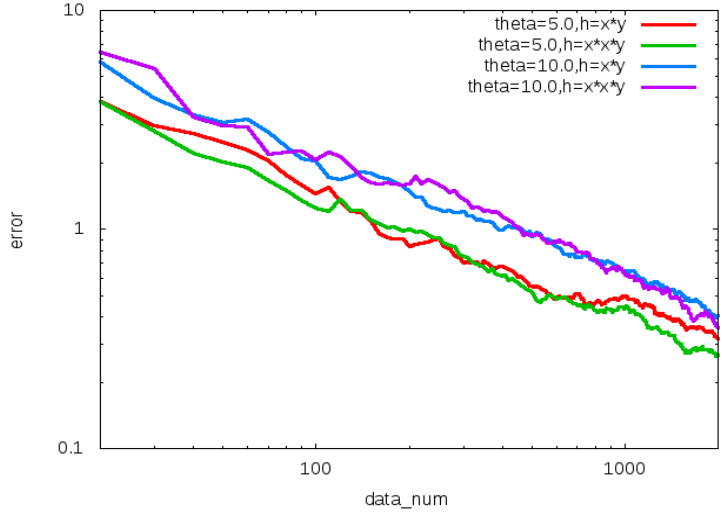}
	\caption{The estimation result of general minimum information copulas. The mean absolute error (vertical axis) with respect to the number of data (horizontal axis) is shown in logarithmic scale. The used parameters are 
	$\theta=5.0,h(x,y)=xy$ (red), 
	$\theta=5.0,h(x,y)=x^2y$ (green),
	$\theta=10.0,h(x,y)=xy$ (blue) and
	$\theta=10.0,h(x,y)=x^2y$ (purple).
	}
	\label{general_copula_log_error}
\end{figure}

The result of linear fitting for the four lines is summarized in Table~\ref{table:linear-fitting-general}.
The proposed method confirms that estimation is possible.

\begin{table}[htbp]
\caption{The coefficients of linear fitting for the four lines on in Figure~\ref{general_copula_log_error}.}
\label{table:linear-fitting-general}
 \begin{center}
 \begin{tabular}{cc|cc}
 $\theta$& $h(x,y)$ & $a$& $b$ \\
 \hline
5.0 & $xy$ & $-0.563269$& 3.02252\\
5.0 & $x^2y$ & $-0.572157$& 2.98715\\
10.0& $xy$& $-0.548056$& 3.30813\\
10.0& $x^2y$& $-0.581245$& 3.52886
\end{tabular}
\end{center}
\end{table}

%Linear fitting is the least-squares fitting of $a,b$ in $y=x^a \exp(b)$ to the data.
% The red line is the estimation error for $\theta=5.0,h(x,y)=xy$ and the result of linear fitting is
% \begin{align*}
% &a=-0.563269\\
% &b=3.02252.
% \end{align*}
% The green line is the estimation error for $\theta=5.0,h(x,y)=x^2y$ and the result of linear fitting is 
% \begin{align*}
% &a=-0.572157\\
% &b=2.98715.
% \end{align*}
% The blue line is the estimation error for $\theta=10.0,h(x,y)=xy$ and the result of linear fitting is 
% \begin{align*}
% &a=-0.548056\\
% &b=3.30813.
% \end{align*}
% The purple line is the estimation error for $\theta=10.0,h(x,y)=x^2y$ and the result of linear fitting is 
% \begin{align*}
% &a=-0.581245\\
% &b=3.52886.
% \end{align*}

% \begin{figure}[htbp]
% 	\centering
% 	\includegraphics[width=10cm,height=7cm]{general_copula_log_error_fitting.png}
% 	\caption{The estimation of general minimum information copulas (logarithmic graph on both axes):\newline
% 		Variation of absolute error (vertical axis) with respect to the number of data (horizontal axis)\newline
% 		red: The linear fitting of $\theta=5.0,h(x,y)=xy$\newline
% 		green: The linear fitting of $\theta=5.0,h(x,y)=x^2y$\newline
% 		blue: The linear fitting of $\theta=10.0,h(x,y)=xy$\newline
% 		purple: The linear fitting of $\theta=10.0,h(x,y)=x^2y$
% 	}
% 	\label{general_copula_log_error_fitting}
% \end{figure}

\section{Conclusions}
In this paper, we propose a generally homogeneous $2$-point local strictly proper score for minimum information copulas.
The greatest strength of this score is that it can be calculated without normalizing functions, which are difficult to compute for minimum information copulas.
The estimator based on the score is asymptotically consistent, and numerical experiments have confirmed that the behavior is consistent with the theory.
Future work includes theoretical computation of asymptotic variance.
%The theoretical guarantee of asymptotic normality is important to know how fast the error variance converges.

In addition, in this paper, we considered the data
\begin{align*}
	(x_1,y_1),\cdots,(x_n,y_n)
\end{align*}
as $N=\lfloor n/2\rfloor$ pairs
\begin{align*}
	(x^1_1,y^1_1,x^2_1,y^2_1),\cdots,(x^1_N,y^1_N,x^2_N,y^2_N).
\end{align*}
However, there are $n(n-1)/2$ pairs in the data: $(x_i,y_i,x_j,y_j)$ for $1\leq i<j\leq n$.
% The number of combinations that pair two data among $N$ data is $\frac{N(N-1)}{2}$ as follows: 
% \begin{align*}
% 	&(x^1,y^1,x^2,y^2),\cdots,(x^1,y^1,x^{N-1},y^{N-1}),(x^1,y^1,x^N,y^N),\\
% 	&(x^2,y^2,x^3,y^3),\cdots,(x^2,y^2,x^N,y^N),\\
% 	&\cdots\\
% 	&(x^{N-1},y^{N-1},x^N,y^N).
% \end{align*}
In terms of extracting the full information of the data, it would be better to consider the empirical score of all combinations
\begin{align*}
	&\hat{S}=\frac{2}{n(n-1)}\sum_{1\leq i < j\leq n}\log\left( \frac{q_{ii}q_{jj}}{q_{ii}q_{jj}+q_{ij}q_{ji}} \right),\quad q_{ij}:=q(x_i,y_j).
\end{align*}
Comparison of the accuracy and computational speed of these two approaches are quite interesting.
The theory of U-statistics (e.g.\ Chapter 12 of \cite{van2000asymptotic}) may be useful for analysis of the modified score.

Furthermore, other than the proposed score, scores with general homogeneity, propriety, and multi-point locality should also be discussed.
In this paper, the score was based on the form of the KL score, but it is not yet known whether a score with similar properties can be created by mimicking the form of the \rm{Hyv\"{a}rinen} score, for example. It is also possible that scores with similar properties can be created using a form that does not take logarithms.
Thus, the authors believe that there is still a way to create such a score.
By knowing the entire set of scores with these properties, we can see if there are scores among them that are even more accurate or that satisfy the properties we want to add.

\section*{Acknowledgments}
We would like to thank Keisuke Yano for helpful comments.
%We would also like to thank Fumiyasu Komaki, Hiromichi Nagao, Teppei Ogiwara, and other members of the Laboratory. 

%%%%%%%%%%%%%%%%%%%%%%%%%%%%%%%%%%%%%%%%%%%%%%
%% Supplementary Material, if any, should   %%
%% be provided in {supplement} environment  %%
%% with title and short description.        %%
%%%%%%%%%%%%%%%%%%%%%%%%%%%%%%%%%%%%%%%%%%%%%%
%\begin{supplement}
%\stitle{???}
%\sdescription{???.}
%\end{supplement}

%\bibliographystyle{myplain}
%\bibliography{reference}       % Bibliography file (usually '*.bib')

\begin{thebibliography}{10}

\bibitem{BHHJ}
A.~Basu, I.~R. Harris, N.~L. Hjort, and M.~C. Jones.
\newblock {Robust and efficient estimation by minimising a density power
  divergence}.
\newblock {\em Biometrika}, 85(3):549--559, 09 1998.

\bibitem{Bedford}
T.~Bedford and K.~Wilson.
\newblock On the construction of minimum information bivariate copula families.
\newblock {\em Annals of the Institute of Statistical Mathematics}, 66, 08
  2014.

\bibitem{10.1214/aos/1176344689}
J.~M. Bernardo.
\newblock Expected information as expected utility.
\newblock {\em The Annals of Statistics}, 7(3):686 -- 690, 1979.

\bibitem{Borwein_et_al1994}
J.~M. Borwein, A.~S. Lewis, and R.~D. Nussbaum.
\newblock Entropy minimization, {DAD} problems, and doubly stochastic kernels.
\newblock {\em J.\ Funct.\ Anal.}, 123:264--307, 1994.

\bibitem{CHATRABGOUN2018266}
O.~Chatrabgoun, A.~Hosseinian-Far, V.~Chang, N.~G. Stocks, and A.~Daneshkhah.
\newblock Approximating non-\protect{G}aussian \protect{B}ayesian networks
  using minimum information vine model with applications in financial
  modelling.
\newblock {\em Journal of Computational Science}, 24:266--276, 2018.

\bibitem{DANESHKHAH2016469}
A.~Daneshkhah, R.~Remesan, O.~Chatrabgoun, and I.~P. Holman.
\newblock Probabilistic modeling of flood characterizations with parametric and
  minimum information pair-copula model.
\newblock {\em Journal of Hydrology}, 540:469--487, 2016.

\bibitem{FUJISAWA20082053}
H.~Fujisawa and S.~Eguchi.
\newblock Robust parameter estimation with a small bias against heavy
  contamination.
\newblock {\em Journal of Multivariate Analysis}, 99(9):2053--2081, 2008.

\bibitem{doi:10.1198/016214506000001437}
T.~Gneiting and A.~E. Raftery.
\newblock Strictly proper scoring rules, prediction, and estimation.
\newblock {\em Journal of the American Statistical Association},
  102(477):359--378, 2007.

\bibitem{hyvarinen2005estimation}
A.~Hyv{\"a}rinen.
\newblock Estimation of non-normalized statistical models by score matching.
\newblock {\em Journal of Machine Learning Research}, 6(Apr):695--709, 2005.

\bibitem{jansen1997maximum}
M.~J. Jansen.
\newblock Maximum entropy distributions with prescribed marginals and normal
  score correlations.
\newblock In {\em Distributions with given marginals and moment problems},
  pages 87--92. Springer, 1997.

\bibitem{10.1214/aoms/1177729694}
S.~Kullback and R.~A. Leibler.
\newblock {On information and sufficiency}.
\newblock {\em The Annals of Mathematical Statistics}, 22(1):79 -- 86, 1951.

\bibitem{Nels06}
R.~B. Nelsen.
\newblock {\em An Introduction to Copulas}.
\newblock Springer, New York, NY, USA, second edition, 2006.

\bibitem{parry2012}
M.~Parry, A.~P. Dawid, and S.~Lauritzen.
\newblock Proper local scoring rules.
\newblock {\em Ann. Statist.}, 40(1):561--592, 02 2012.

\bibitem{Sei&Yano}
T.~Sei and K.~Yano.
\newblock Minimum information dependence modeling for arbitrary product spaces.
\newblock in preparation, 2022.

\bibitem{van2000asymptotic}
A.~W. V.~d. Vaart.
\newblock {\em Asymptotic Statistics}.
\newblock Cambridge University Press, 2000.

\end{thebibliography}

%% or include bibliography directly:
% \begin{thebibliography}{}
% \bibitem{b1}
% \end{thebibliography}

\end{document}